\documentclass[5p]{elsarticle}
\usepackage{multirow}
\usepackage{multicol}
\usepackage{array}

\usepackage{booktabs}
\usepackage{natbib,hyperref}
\usepackage{balance}
\usepackage{listings}
\usepackage{tabulary}
\usepackage{setspace}

\usepackage{color}
\usepackage{soul}
\definecolor{myblue}{rgb}{0,0,0.9}
\definecolor{mygray}{rgb}{0.9,0.9,0.9}
\definecolor{mymauve}{rgb}{0.58,0,0.82}
\usepackage[cmex10]{amsmath}
\usepackage{graphicx}
\usepackage{amssymb}
\usepackage{stmaryrd}
\usepackage{algorithm,algorithmic}
\usepackage{amsthm}
\usepackage{wrapfig}

\newcommand{\ALOOP}[1]{\ALC@it\algorithmicloop\ #1%
  \begin{ALC@loop}}
\newcommand{\ENDALOOP}{\end{ALC@loop}\ALC@it\algorithmicendloop}

\newtheorem{theorem}{\textbf{\emph{Theorem}}}

\newcommand{\main}{VPFedMF}

\definecolor{garrisonpink1}{rgb}{0.858, 0.188, 0.478}

\definecolor{captwan}{rgb}{0.056, 0.318, 0.958}

\journal{Knowledge-Based Systems}

\begin{document}

\begin{frontmatter}

\title{Towards Privacy-Preserving and Verifiable Federated Matrix Factorization}

\author[els,rng]{Xicheng Wan}
\ead{xicheng.wan@outlook.com}

\author[focal]{Yifeng Zheng\corref{cor1}}
\ead{yifeng.zheng@hit.edu.cn}

\author[rvt]{Qun Li}
\ead{120106222757@njust.edu.cn}


\author[rvt]{Anmin Fu}
\ead{fuam@njust.edu.cn}

\author[rvt]{Mang Su}
\ead{sumang@njust.edu.cn}

\author[rvt]{Yansong Gao}
\ead{yansong.gao@njust.edu.cn}


\cortext[cor1]{Corresponding author}

\address[els]{School of Automation,
Nanjing University of Science and Technology, Nanjing, JiangSu, China}
\address[rng]{Department of Information Engineering, The Chinese University of Hong Kong, Hong Kong, China}
\address[focal]{School of Computer Science and Technology,
Harbin Institute of Technology, Shenzhen, Guangdong, China.}
\address[rvt]{School of Computer Science and Engineering,
Nanjing University of Science and Technology, Nanjing, JiangSu, China}



\begin{abstract}
Recent years have witnessed the rapid growth of federated learning (FL), an emerging privacy-aware machine learning paradigm that allows collaborative learning over isolated datasets distributed across multiple participants. The salient feature of FL is that the participants can keep their private datasets local and only share model updates. Very recently, some research efforts have been initiated to explore the applicability of FL for matrix factorization (MF), a prevalent method used in modern recommendation systems and services. It has been shown that sharing the gradient updates in federated MF entails privacy risks on revealing users’ personal ratings, posing a demand for protecting the shared gradients. Prior art is limited in that they incur notable accuracy loss, or rely on heavy cryptosystem, with a weak threat model assumed. In this paper, we propose VPFedMF, a new design aimed at privacy-preserving and verifiable federated MF. VPFedMF provides guarantees on the confidentiality of individual gradient updates through lightweight and secure aggregation. Moreover, VPFedMF ambitiously and newly supports correctness verification of the aggregation results produced by the coordinating server in federated MF. Experiments on a real-world movie rating dataset demonstrate the practical performance of VPFedMF in terms of computation, communication, and accuracy.

\end{abstract}


\end{frontmatter}

\section{Introduction}

Privacy-preserving machine learning has been gaining increasing attentions from both academia and industry (e.g., Google and WeBank) in recent years because of the increased user privacy awareness in society and enforcement of data privacy laws such as the General Data Protection Regulation (GDPR, effective in May 2018)~\cite{gdpr}, California Privacy Rights Act (CPRA, effective in Jan. 2021)~\cite{cpra}, and China Data Security Law (CDSL, effective in Sep. 2021) \cite{cdsl}. 
Federated learning (FL) is one of the most popular paradigms in recent years for providing privacy protection in machine learning~\cite{mcmahan2017communication,li2020federated,zhang2021survey,zheng22}, and has demonstrated applicability for various application scenarios ranging from resource-limited mobile devices~\cite{gao2020end} to resource-rich institutions, e.g., medical centers~\cite{xu2021federated}. 
In FL, the participants can keep their private datasets locally, yet are able to train a global model over the joint datasets \cite{li2021survey}. 
A centralized server coordinates the participants and aggregates their local model updates (instead of their raw private datasets) to iteratively update the global model.

The FL paradigm has seen successful applications in scenarios that deal with privacy-sensitive data.
For example, in financial systems like open banking \cite{long2020federated}, FL can be leveraged to identify malicious clients with act of loan swindling and escaping from paying for the debt without exposing all clients' financial information \cite{yang2019federated}. 
On the other hand, it is noted that most existing FL systems and services have mainly focused on deep neural networks \cite{nasr2019comprehensive,perifanis2022federated,wang2020optimizing}.
Very recently, only few research efforts have been initiated to explore the applicability of FL for matrix factorization \cite{koren2009matrix}, a prevalent method that has seen wide use in recommendation systems for rating prediction, item ranking, item recommendation, and more \cite{yu2017attributes,zhang2019probabilistic,yang2021fcmf}.
Generally, MF decomposes a user-item rating matrix into two latent representations or components: a user profile matrix and an item profile matrix, where a new prediction can be made with the combination of both matrices.

The conventional MF is performed in a centralized manner, which may easily cause violation of data privacy.
Indeed, user ratings contains private information such as user behavior, preferences and social status~\cite{kosinski2013private}. Therefore, it is imperative to protect user privacy in MF while making quality recommendations. There are efforts towards addressing this concern when the MF is trained in a centralized manner. Berlioz \textit{et al.} \cite{berlioz2015applying} propose to utilize differential privacy \cite{DworkMNS06} to obfuscate users' raw data for the sake of securing model results after training by a centralized model with a trade-off of accuracy loss. Some works \cite{nikolaenko2013privacy,kim2016efficient,bellare2012foundations} resort to cryptographic techniques (like powerful yet expensive homomorphic encryption and garbled circuits).
These works, however, still all fall within centralized training settings and lack scalability for practical deployment.

Until very recently, Chai \textit{et al.}~\cite{chai2020secure} initiate the study on how to bridge FL and MF, enabling MF to be conducted in a FL setting.
MF in a FL setting updates the user profile matrix only at the user side while aggregating gradient information and updating the item profile matrix at the server side.
This considers the fact that the user profile matrix encodes private preference information.
In this context, Chai \textit{et al.} analyze the privacy leakage in the context of federated MF and find that user rating information could still be leaked when the server can see and analyze the gradient information uploaded by the users.
As a solution, they apply additive homomorphic encryption (AHE) to protect the gradient information in aggregation and propose a design called FedMF.
Despite that FedMF neither requires raw datasets from users nor leaks the gradient information through the use of AHE, it incurs significant performance overheads.
%
Moreover, FedMF works under a relatively weak security model, and does not offer assurance on the computation integrity of aggregation against the server.

In light of the above, this work proposes VPFedMF, a new protocol for enabling privacy-preserving and verifiable matrix factorization.
VPFedMF protects the confidentiality of gradient information of individual users throughout the whole process of federated matrix factorization, through an advanced \textit{masking-based secure aggregation} technique with low overhead. 
In particular, in VPFedMF, users can provide encrypted gradient information through lightweight encryption, while the server is still able to perform aggregation of the encrypted gradient updates.
This is in substantial contrast to the state-of-the-art work \cite{chai2020secure} which relies on the usage of heavy homomorphic cryptosystem.
In the meantime, VPFedMF newly and ambitiously provides \textit{assurance on the integrity of aggregation} against the server, achieving much stronger security than \cite{chai2020secure}.
In particular, VPFedMF introduces a delicate verification mechanism that allows users to verify the correctness of the aggregation result received from the server in each iteration.
An adversarial server that does not correctly perform the aggregation would be detected. 
We highlight our contributions as follows.

\begin{itemize}
\item We present a new protocol VPFedMF, which provides cryptographic guarantees on the confidentiality of gradient information of individual users in federated matrix factorization, through masking-based lightweight and secure aggregation.

\item VPFedMF newly provides assurance on the integrity of aggregation against the server, under a stronger threat model that was overlooked by prior work.
Through a delicate cryptographic verification mechanism, VPFedMF allows user-side verification of the correctness of aggregation results produced by the server. 

\item We make an implementation of VPFedMF and perform a thorough performance evaluation on a real-world movie rating dataset MovieLens.
Compared with the state-of-the-art work FedMF \cite{chai2020secure}, VPFedMF is about  $20\times$ faster.
Experiments also validate that VPFedMF preserves the accuracy, matching that of plaintext-domain federated MF and conventional centralized MF.

\end{itemize}

The rest of the paper is organized as below. Section~\ref{sec:preliminary} provides necessary preliminaries. Section~\ref{sec:VPFedMF} elaborates on our system model, threat model, and the detailed construction, followed by the security analysis in Section~\ref{sec:security analysis}. Section~\ref{sec:experiment} provides the performance evaluation and comparison. Section~\ref{sec:conclusion} concludes the whole paper.

\section{Technical Preliminaries}\label{sec:preliminary}

This section provides preliminaries related to the construction of VPFedMF. We firstly introduce matrix factorization in a federated learning setting. Then we describe several cryptographic primitives to be used later. 

\subsection{Federated Matrix Factorization}

The MF \cite{koren2009matrix}, \cite{takacs2008investigation}, \cite{gemulla2011large} technique has been popularly used in recommendation systems. 
%
Given a sparse rating matrix $\mathbf{R}\in \mathbb{R}^{n\times m}$, MF aims to generate a user profile matrix $\mathbf{U}\in \mathbb{R}^{n\times d}$ and an item profile matrix $\mathbf{V}\in \mathbb{R}^{m\times d}$ with the same latent dimension $d$, where $n$ is the number of users and $m$ is the number of items.
The $i$-th row of $\mathbf{U}$ represents the profile of the $i$-th user $\mathcal{U}_i$, and the $k$-th row of $\mathbf{V}$ represents the profile of the $k$-th item $\mathcal{V}_k$.
Let $r_{i,k}$ denote the rating value generated by user $\mathcal{U}_i$ for item $\mathcal{V}_k$.
The resulting matrices $\mathbf{U}$ and $\mathbf{V}$ after training can then be used to generate predictions $r_{i,k}'$ for the rating values for all user/item pairs, i.e., $r_{i,k}'=\langle \mathbf{u}_i,\mathbf{v}_k \rangle$, where $\mathbf{u}_i\in \mathbb{R}^{d}$ is the profile vector for user $\mathcal{U}_i$ and $\mathbf{v}_k\in \mathbb{R}^{d}$ is the profile vector for item $\mathcal{V}_k$.

The computation of the user profile matrix $\mathbf{U}$ and item profile matrix $\mathbf{V}$ can be achieved by solving the following regularized least squares minimization problem:
\begin{equation*}\label{lossfunction}    
    \mathop{\arg\min}\limits_{\mathbf{U,V}}\frac{1}{M}\sum\nolimits_{(i,k) \in \Omega}(r_{i,k}-\langle \mathbf{u}_i,\mathbf{v}_k \rangle)^2+\lambda||\mathbf{U}||^2_2+\mu||\mathbf{V}||^2_2,
\end{equation*}
where $M$ is the total number of ratings, $ \Omega\subseteq\{1,2,\dots,n\}\times\{1,2,\dots,m\}$ is a set for indices pairs $(i,k)$ and $|\Omega| = M$. $\lambda$ and $\mu$ are small positive values in order to avoid overfitting.
To solve this optimization problem, the method of stochastic gradient descent (SGD) is usually applied, which iteratively updates $\mathbf{U}$ and $\mathbf{V}$ through the following rules in an iteration $t$:
\begin{equation*}
    \mathbf{u}_i^{t}=\mathbf{u}_i^{t-1}-\mathbf{H}_{i}^t\label{update1};
\end{equation*}
\begin{equation*}
    \mathbf{v}_k^{t}=\mathbf{v}_k^{t-1}-\mathbf{G}_{k}^t\label{update2},
\end{equation*}
where $\mathbf{H}_{i}^t$ and $\mathbf{G}_{k}^t$ are gradient vectors that are computed based on the current user profile matrix $\mathbf{U}^{t-1}$ and item profile matrix $\mathbf{V}^{t-1}$, as shown below:
\begin{equation*}\label{2}
    \mathbf{H}_{i}^t=\sum\nolimits_{k\in [1,m]}\gamma[-2\mathbf{v}_k^{t-1}(r_{i,k}-\langle \mathbf{u}_i^{t-1},\mathbf{v}_k^{t-1}\rangle)+2\lambda \mathbf{u}_i^{t-1}];
\end{equation*}
where $\gamma$ is also a small positive value to control the convergence speed. $\mathcal{U}_i$ generates the gradient vector $\mathbf{G}_{i, k}^t$ for each item $\mathcal{V}_k$:
\begin{equation*}\label{3}
    \mathbf{G}_{i, k}^t=\gamma[-2\mathbf{u}_i^{t-1}(r_{i,k}-\langle \mathbf{u}_i^{t-1},\mathbf{v}_k^{t-1}\rangle)+2\mu \mathbf{v}_k^{t-1}].
\end{equation*}
Then we have
\begin{equation*}\label{4}
    \mathbf{G}_{k}^t=\sum\nolimits_{i\in[1,n_k]} \mathbf{G}_{i, k}^t,
\end{equation*}
where $n_k$ is the number of users providing ratings for item $\mathcal{V}_k$.
Conventionally, MF is performed in a centralized setting where all the ratings are collected by a server for processing.
Recently, there have been research efforts on supporting MF in a distributed manner, particularly using the FL paradigm, for the purpose of reducing privacy risks by avoiding the exposure of raw rating values \cite{chai2020secure}. 
The process of federated MF is detailed in the Algorithm \ref{DistributedMF}. 
It is executed between a server and a set of users that hold their rating values locally. 
In each iteration $t$, the server sends the current item profile matrix $\mathbf{V}^{t-1}$ to all users.
Note that in the first iteration, the server initializes $\mathbf{V}^{0}$ and each user $\mathcal{U}_i$ generates its user vector $\mathbf{u}^0_i$.
Given $\mathbf{V}^{t-1}$, each user $\mathcal{U}_i$ computes the gradient vector $\mathbf{H}_{i}^t$, which is used to update the user vector $\mathbf{u}^{t}_i$.
Each user $\mathcal{U}_i$ then computes a gradient vector $\mathbf{G}_{i, k}^t$ for each item $\mathcal{V}_k$ based on its ratings and the vector $\mathbf{v}^{t-1}_k$ derived from $\mathbf{V}^{t-1}$.
Each user $\mathcal{U}_i$ uploads its gradient vector $\mathbf{G}_{i, k}^t$ to the server, which aggregates these gradient vectors and produces an aggregate gradient vector $\mathbf{G}_{k}^t=\displaystyle \sum\nolimits_{i\in[1,n_k]} {\mathbf{G}_{i, k}^t}$.
The aggregate gradient vector is used to update the item vector $\mathbf{v}^{t}_k$, through $\mathbf{v}^{t}_k=\mathbf{v}^{t-1}_k-\mathbf{G}_{k}^t$.

While performing MF under the federated learning par$\-$adigm avoids the sharing of raw ratings, the sharing of gradients has been shown to be subject to attacks which could infer the rating values, compromising the data privacy \cite{chai2020secure}.
Hence, it is necessary to offer protection on the shared gradients in FedMF.

\begin{algorithm}[!t]
\caption{Federated MF in the Plaintext Domain}
\label{DistributedMF}
\begin{algorithmic}[1]

\REQUIRE Initialized user vector $\mathbf{u}^0_i$ on the user side and item matrix $\mathbf{V}^0$ on the server side.
\ENSURE Trained user matrix $\mathbf{U}$ and item matrix $\mathbf{V}$.

\FOR{each iteration $t=1,2,\cdots$}

\STATE Users download latest item profile matrix $\mathbf{V}^{t-1}$ from the server.

\FOR{each user $\mathcal{U}_i$} 
    \STATE Compute gradient $\mathbf{H}_{i}^t$.
\STATE Compute $\mathbf{u}_i^{t}=\mathbf{u}_i^{t-1}-\mathbf{H}_{i}^t$.
\STATE Compute $\mathbf{G}_{i,k}^t$ for each item $\mathcal{V}_k$.
\STATE Send $\mathbf{G}_{i,k}^t$ to the server.
\ENDFOR

\STATE The server aggregates all $\mathbf{G}_{i,k}^t$ for each item $\mathcal{V}_k$ to produce $\mathbf{G}_{k}^t$. 

\STATE The server updates the item vectors: $\mathbf{v}^t_k =\mathbf{v}^{t-1}_k - \mathbf{G}_{k}^{t}$.

\ENDFOR

\end{algorithmic}

\end{algorithm}

\subsection{Homomorphic Hash Function}
Homomorphic hash function $\mathsf{HF}(\cdot)$ enables to compress a vector by computing a hash of the vector, while preserving the addition property \cite{bellare1994incremental}.
It is based on the hardness of the discrete logarithm in groups of prime order.
%
%
Let $\mathbb{G}$ denote a cyclic group of prime order $q$ with generator $g$, and $g_1,...,g_{d}$ represent distinct elements randomly chosen from $\mathbb{G}$.
Given a $d$-dimensional vector $\mathbf{x}$, which the $l$-th element is denoted by $x_l$, the homomorphic hash $h_{\mathbf{x}}$ of $\mathbf{x}$ is computed via $$h_{\mathbf{x}}=\textsf{HF}(\mathbf{x})=\prod\nolimits_{l \in [1,d]} g_l^{x_l}.$$

\subsection{Commitment}

A commitment scheme allows one to commit to a message ahead of time \cite{damgaard1998commitment}.
Later, the message is revealed, and the commitment can be used to check whether the revealed message is indeed the one committed in the beginning.
%
%
A secure commitment scheme guarantees that a message cannot be modified after being committed.
Besides, the commitment can hide the underlying committed message.
A commitment scheme proceeds in two phases: the commit phase and the decommit phase.
In the commit phase, a commitment for a message $\mathcal{M}$ is generated by $c = \textsf{Commit} (\mathcal{M};r)$, where $r$ is randomness.
In the decommit phase, a message $\mathcal{M}'$ is revealed, and a function $\textsf{DeCommit} (\mathcal{M}',c,r)$ is run to check whether $\mathcal{M}'$ is the message underlying the commitment $c$.
The function $\textsf{DeCommit}(\cdot)$ outputs $1$ which indicates successful verification or $0$ indicating the verification failure.

%

\section{VPFedMF}\label{sec:VPFedMF}

\subsection{Overview}


The overview of our proposed VPFedMF system framework is illustrated in Fig. \ref{fig:system_architecture}.
VPFedMF enables matrix factorization in a federated learning setting, while preventing privacy leakages from the gradients by aggregating gradients in the ciphertext domain via secure aggregation techniques.
In the meantime, it aims to enforce that the (secure) aggregation is correctly conducted by the server through the integration of a verification mechanism.
We elaborate on the design rationale as follows.

\begin{figure}[t!]
\centerline{\includegraphics[width=0.48\textwidth]{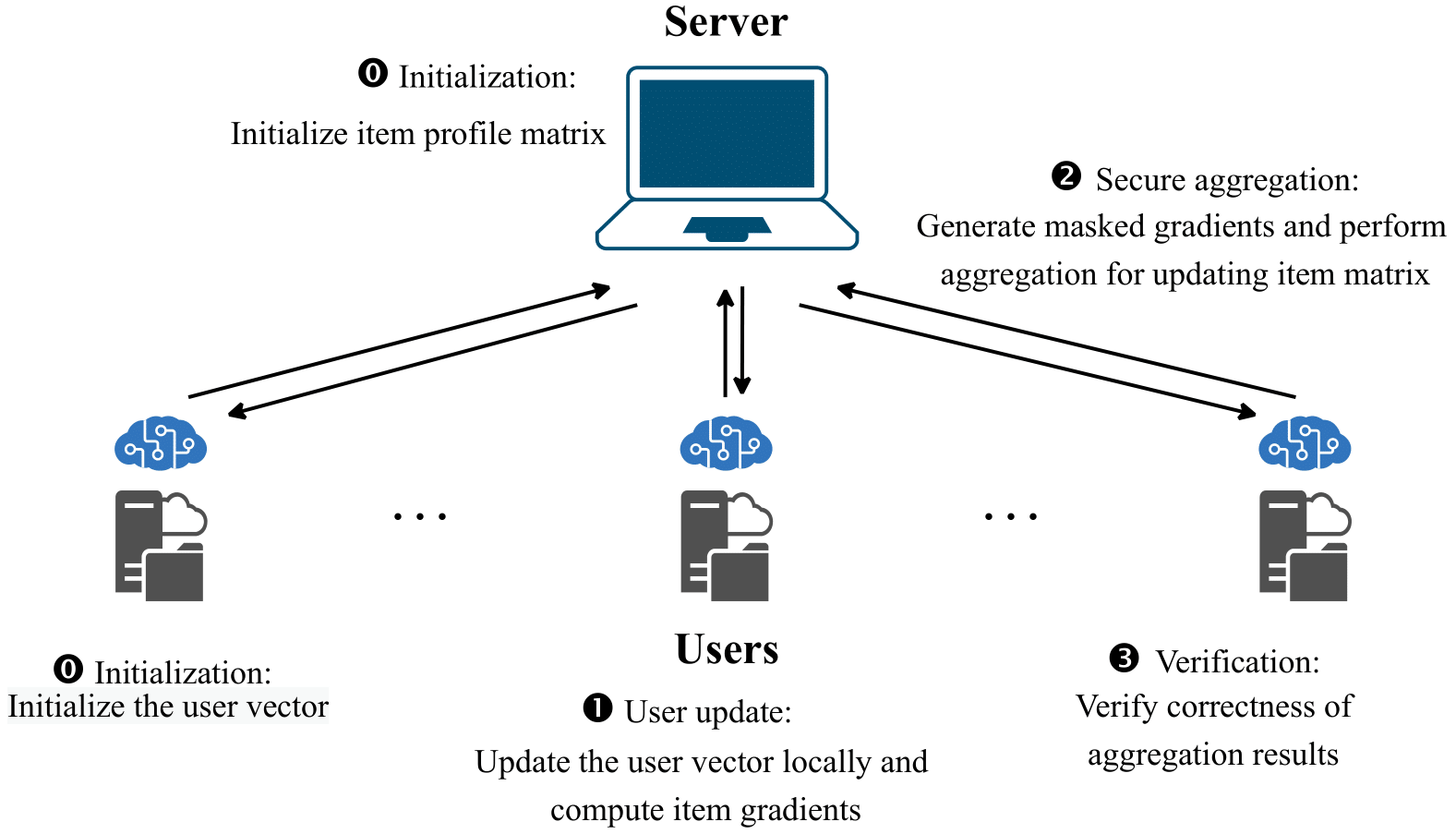}}
\caption{The system overview of our proposed VPFedMF design.}
\label{fig:system_architecture}
\end{figure}

Unlike the prior work \cite{chai2020secure} that relies on heavy homomorphic encryption for secure aggregation, VPFedMF resorts to a newly developed masking-based lightweight secure aggregation technique \cite{bonawitz2017practical} for encrypting each user's gradient vector while supporting aggregation of the encrypted gradient vectors.

Specifically, in VPFedMF, each user $\mathcal{U}_i$ will generate a tailored random masking vector for encrypting the gradient vector $\mathbf{G}_{i, k}^t$ for each item $\mathcal{V}_k$.
The random masking vector is generated based on each user's secret key and public keys of other users in the system.
And the generation process only requires the usage of a pseudo-random number generator and thus is fast compared with homomorphic encryption.
Once the random masking vector is generated, encryption is achieved via fast modulo addition.
In order to guarantee the integrity of aggregation result which could be potentially corrupted by the server, we take advantage of cryptographic techniques including homomorphic hash function and commitment to foster a verification mechanism in VPFedMF, inspired by the recent work \cite{guo2020v}.
Specifically, before sending the encrypted gradient vector to the server in an iteration, each user first commits to its gradient vector based on the homomorphic hash function and commitment scheme.
The commitments are sent to all other users in the system, which will be used later to verify the integrity of the aggregation result received from the server.
Based on the above insights, this paper presents the first design for verifiable and privacy-preserving federated matrix factorization.


\begin{figure*}[!t]
\centering

\fbox{
  \begin{minipage} [t]{0.99\textwidth}

\textbf{Initialization (Phase 0, only once in iteration $1$):}

\begin{enumerate}
\setcounter{enumi}{-1}

\item \textit{Key generation:}

User: Each $\mathcal{U}_i$ generates its private key $\textrm{msk}_i$ and public key $\textrm{mpk}_i$. $\mathcal{U}_i$ sends its public key $\textrm{mpk}_i$ to server.

Server: Server receives public key $\textrm{mpk}_i$ from $\mathcal{U}_i$ and broadcasts it to other $\mathcal{U}_j$.

User: Each $\mathcal{U}_i$ computes its shared key $\textrm{ck}_{i,j} = \textsf{KeyAgreement}(\textrm{msk}_i, \textrm{mpk}_j)$ with respect to another user $\mathcal{U}_j$.

\item \textit{Profile initialization:}

User: Each $\mathcal{U}_i$ initializes the user vector $\mathbf{u}_i^0$. 

Server: Server prepares for initial item profile matrix $\mathbf{V}^0$.

\end{enumerate}

\textbf{User Update (Phase 1):}

Each $\mathcal{U}_i$ receives the latest item profile matrix $\mathbf{V}^{t-1}$ from the server.
Then, $\mathcal{U}_i$ computes the user gradient vector $\mathbf{H}_{i}^t$ and the item gradient vector $\mathbf{G}_{i,k}^t$ for each item $\mathcal{V}_k$.
The gradient vector $\mathbf{H}_{i}^t$ is used to update locally the user profile vector $\mathbf{u}_i^{t-1}$ to $\mathbf{u}_i^{t}$, via $\mathbf{u}_i^{t}=\mathbf{u}_i^{t-1}-\mathbf{H}_{i}^t$. 
The gradient vector $\mathbf{G}_{i,k}^t$ for each item will enter the next secure aggregation phase.

\textbf{Secure Aggregation (Phase 2):}

\begin{enumerate}
\setcounter{enumi}{-1}
\item \textit{Making commitments:} 

User: Each $\mathcal{U}_i$ computes $h_{i,k}^t = \textsf{HF}(\displaystyle \frac{1}{n_k}\mathbf{v}_k^{t-1}-\mathbf{G}_{i,k}^t)$ and $c_{i,k}^t = \textsf{Commit}(h_{i,k}^t;r^t_{i,k})$ for item $\mathcal{V}_k$. $\mathcal{U}_i$ sends its commitment $c_{i,k}^t$ to the server.

Server: The server receives  $c_{i,k}^t$ from $\mathcal{U}_i$ and broadcasts it to other users $\mathcal{U}_j$. 

\item \textit{Masking gradient vectors:}

User: Each $\mathcal{U}_i$ expands $\textrm{ck}_{i,j}$ by applying a pseudo-random number generator (\textsf{PRNG}) and $\Delta$ to a $d$-dimensional vector for masking. 
In particular, each $\mathcal{U}_i$ computes  $\boldsymbol{\sigma}_{i,k}^t = (\displaystyle \frac{1}{n_k}\mathbf{v}_k^{t-1}-\mathbf{G}_{i,k}^t) + \sum\nolimits_{j\in [1,n_k] \backslash \{i\}} \Delta_{i,j}\textsf{PRNG}(\mathrm{ck}_{i,j}||k||t) \bmod B$ for item $\mathcal{V}_k$, where $\Delta_{i,j} = 1$ if $i<j$ and $\Delta_{i,j} = -1$ if $i>j$, and $B$ is a modulus defining the message space.
Each $\mathcal{U}_i$ sends $\boldsymbol{\sigma}_{i,k}^t$ to the server.

\item \textit{Aggregating masked gradient vectors:}

Server: The server receives $\boldsymbol{\sigma}_{i,k}^t$ from users, and computes the aggregation result $\mathbf{v}_k^t =\displaystyle \sum\nolimits_{i\in[1,n_k]} \boldsymbol{\sigma}_{i,k}^t \bmod B$ for all items, where it is derived that $\mathbf{v}_k^t=\mathbf{v}_k^{t-1}-\mathbf{G}_{k}^t$. The server broadcasts the aggregation result of each item (i.e., the updated $\mathbf{V}^t$) to all users. 
\end{enumerate}

\textbf{Verification (Phase 3):}
\begin{enumerate}
\setcounter{enumi}{-1}

\item \textit{Decommitting:}

User: Each $\mathcal{U}_i$ sends to the server its decommitment strings, i.e., hashes and corresponding randomnesses $\{h_{i,k}^t,r^t_{i,k}\}$.

Server: The server receives $\{h_{i,k}^t,r^t_{i,k}\}$ from $\mathcal{U}_i$ and broadcasts it to other users $\mathcal{U}_j$.

\item \textit{Commitment verification:}

User: Each $\mathcal{U}_i$ first checks for each item whether the received decommitment strings $\{h_{j,k}^t,r^t_{j,k}\}$ of all other users can pass a commitment verification, via checking whether $1 \mathop  = \limits^? \textsf{DeCommit}(h_{j,k}^t,c_{j,k}^t,r^t_{j,k})$, for each $j\in[1,n_k]\backslash \{i\}$.
If the equality test holds for every $j$ and every $k$, $\mathcal{U}_i$ moves to the next \textit{Aggregation Result Verification} step. Otherwise, $\mathcal{U}_i$ outputs $\perp$ and abort.

\item \textit{Aggregation result verification:}

User: Each $\mathcal{U}_i$ checks the integrity of the aggregation result $\mathbf{v}_k^t$ for each item $\mathcal{V}_k$ through the following equality test: \textsf{HF}($\mathbf{v}_k^t) \mathop  = \limits^?\displaystyle \prod\nolimits_{i\in[1,n_k]} h_{i,k}^t$. 
If the equality holds for all items, $\mathcal{U}_i$ accepts the updated item matrix $\mathbf{V}^{t}$ and moves to next iteration. 
Otherwise, $\mathcal{U}_i$ outputs $\perp$ and abort.

\end{enumerate}

\end{minipage}
}

\caption{The full protocol of VPFedMF (in an iteration $t$).}
\label{fig:Privacy security protocol}
\end{figure*} 

\subsection{Threat Model}

In \main, we consider that the server may be compromised by an adversary.
%
%
The adversary may attempt to infer the private gradient vectors of users, threatening the confidentiality of the raw rating values held by users locally.
Besides, the adversary may instruct the server to not correctly perform the aggregation over the gradient vectors received from users in each iteration, threatening the integrity of aggregation result for matrix factorization.
In addition, we consider that the adversary may corrupt a subset of users and know their gradient vectors.
%
%
Our security goal is to ensure the confidentiality of individual honest users' gradient vectors against other parties in the system as well as the integrity of the aggregation result against the server, throughout the whole VPFedMF procedure.
As a standard and basic assumption for secure systems \cite{ChaudhariRS20,Eskandarian22}, we assume the interactions among all parties are established via encrypted and authenticated communication channels realized via the Transport Layer Security (TLS) protocol.

\subsection{Detailed Construction}

The proposed protocol in VPFedMF for verifiable and privacy-preserving federated matrix factorization is detailed in Fig.~\ref{fig:Privacy security protocol}.
The protocol proceeds in four phases: \textit{Initialization, User Update, Secure Aggregation,} and \textit{Verification}. 
The initialization phase is performed only once at the start of the protocol, while the other three phases run sequentially in an iteration.
In what follows, we introduce the processing in each phase.
It is noted that in our protocol each user works in parallel when uploading (encrypted) data to the server. And for simplicity of presentation, we focus on introducing the processing on user $\mathcal{U}_i$.

\subsubsection{Initialization}

At the beginning, based on $\textsf{KeyAgreement}$ scheme \cite{goldreich2009foundations}, each user $\mathcal{U}_i$ generates a key pair $(\textrm{msk}_i,$ $ \textrm{mpk}_i)$ using the same group $\mathbb{G}$ with prime order $q$ and generator $g$ in $\mathsf{HF}(\cdot)$, where $\textrm{msk}_i$ is the secret key randomly chosen from $\mathbb{Z}_q$ and $\textrm{mpk}_i$ is the public key which is computed by $\textrm{mpk}_i=\textrm{msk}_i \cdot g$.
Then each $\mathcal{U}_i$ sends the public key $\textrm{mpk}_i$ to the server, which then broadcasts it to other users in the system.
%
%
Each $\mathcal{U}_i$ initializes its vector $\mathbf{u}_i^0$ and generates corresponding shared key $\textrm{ck}_{i,j}=\textrm{msk}_i \cdot \textrm{mpk}_j$ with other users' public key $\textrm{mpk}_j$, which is denoted as $\textrm{ck}_{i,j}=\textsf{KeyAgreement}(\textrm{msk}_i, \textrm{mpk}_j)$. 
The server initializes the item profile matrix $\mathbf{V}^0$.
It is noted that the key generation and distribution process are one-off and performed offline, which do not affect the online system performance.

\subsubsection{User Update}

In the $t$-th iteration, each user $\mathcal{U}_i$ generates two types of gradient vectors: (i) $\mathbf{H}_{i}^t$ for itself and (ii) $\mathbf{G}_{i,k}^t$ for the item vector update. $\mathbf{H}_{i}^t$ is utilized for updating the corresponding user vector $\mathbf{u}_i^{t-1}$ to produce $\mathbf{u}_i^{t}$, while item vectors $\mathbf{v}_k^{t-1}$ from the last iteration and $\mathbf{G}_{i,k}^t$ will be adequately encrypted and submitted to the server in the next phase.

\subsubsection{Secure Aggregation}

In this phase, secure aggregation is performed to securely aggregate each item $\mathcal{V}_k$'s gradient vectors collected from the users, so as to produce an updated vector for each item $\mathcal{V}_k$.
Besides, in order to simultaneously ensure the integrity of the aggregation at the server side, VPFedMF also enforces a verification mechanism based on the cryptographic techniques including commitment and homomorphic hash function, as mentioned above.
The secure aggregation phase in VPFedMF runs as follows.

Firstly, each $\mathcal{U}_i$ calculates $\displaystyle\frac{1}{n_k}\mathbf{v}_k^{t-1}-\mathbf{G}_{i,k}^t$ for each item $\mathcal{V}_k$, which will serve as its input in the secure aggregation.
Then, each $\mathcal{U}_i$ generates commitments $c_{i,k}^t$ for its inputs, through: $ h^t_{i,k}=\textsf{HF}(\displaystyle\frac{1}{n_k}\mathbf{v}_k^{t-1}-\mathbf{G}_{i,k}^t)$, and $ c^t_{i,k}=\textsf{Commit}( h^t_{i,k};r^t_{i,k})$.

Each $\mathcal{U}_i$ then sends the commitments $\{c_{i,k}^t\}$ to the server, which then broadcasts them to other users.
The input messages $\{h_{i,k}^t\}$ and randomnesses $\{r_{i,k}^t\}$ to the commitments are kept locally.
Subsequently, each $\mathcal{U}_i$ generates an encrypted gradient vector based on lightweight masking.
In particular, each $\mathcal{U}_i$ computes a masking vector from the shared key $\textrm{ck}_{i,j}$, based on the delicate use of a pseudo-random number generator ($\textsf{PRNG}$), as seen in Fig.~\ref{fig:Privacy security protocol}.
The way of mask generation ensures that the masking vectors will cancel out once the sum of masked gradient vectors are formed.
After performing the random masking (i.e., step 1 in the part of secure aggregation in Fig. \ref{fig:Privacy security protocol}), each $\mathcal{U}_i$ produces $\boldsymbol{\sigma}_{i,k}^t$, which is sent to the server.
It is noted that due to the enforcement of random masking, $\boldsymbol{\sigma}_{i,k}^t$ is indistinguishable from a vector filled with random values. So the server cannot infer the original data.
Upon receiving $\boldsymbol{\sigma}_{i,k}^t$ from users, the server computes the aggregation result $\mathbf{v}_k^t$ by summing up the masked vectors.
The server then broadcasts the aggregation result of each item (i.e., the updated $\mathbf{V}^t$) to all users.


\subsubsection{Verification}

This phase runs when each $\mathcal{U}_i$ receives the updated item matrix $\mathbf{V}^{t}$ from the server. 
At the beginning, each $\mathcal{U}_i$ sends the commitment inputs $\{h_{i,k}^t,r^t_{i,k}\}$ to the server, which then forwards them to other users $\mathcal{U}_j$.
Next, each $\mathcal{U}_i$ proceeds in a two-step verification process.
Firstly, $\mathcal{U}_i$ performs a commitment verification for each $j\in[1,n_k]\backslash \{i\}$:
$$1 \mathop  = \limits^? \textsf{DeCommit}(h_{j,k}^t,c_{j,k}^t,r^t_{j,k})$$
If the equality does not hold any $j$, $\mathcal{U}_i$ outputs $\perp$ and aborts.
Otherwise, $\mathcal{U}_i$ moves on to the next step for verifying the integrity of the aggregation result.
In particular, $\mathcal{U}_i$ performs the following equality test: $$\textsf{HF}(\mathbf{v}_k^t) \mathop  = \limits^? \displaystyle\prod\nolimits_{i\in[1,n_k]} h_{i,k}^t$$ 
If the equality holds for all items, $\mathcal{U}_i$ accepts the updated item matrix $\mathbf{V}^{t}$ and moves to next iteration. 
Otherwise, $\mathcal{U}_i$ outputs $\perp$ and aborts.

\subsection{Remarks}
The presented VPFedMF design not only preserves the confidentiality of items' gradient information from users but also provides strong verification for the aggregation results.
In comparison with the prior art \cite{chai2020secure} that relies on heavy homomorphic encryption for encrypting gradient vectors and supporting privacy-preserving aggregation, VPFedMF newly resorts to lightweight masking-based cryptographic techniques for protecting the privacy of gradient vectors in aggregation. 
For privacy-preserving aggregation, users only need to perform some lightweight hashing operations and arithmetic operations.
For verifiability, the use of homomorphic hash function allows to greatly compress the high-dimensional vectors into constant-sized elements, facilitating the computation of commitments.
The security of homomorphic hash function and commitment ensures that the underlying plaintext gradient vectors of an individual user are strongly protected against the server and other users in the system.

\section{Security Analysis}\label{sec:security analysis}

VPFedMF guarantees the integrity of the aggregation as well as individual user privacy. Hereafter, we analysis its security to justify the security guarantees. To ease the description, we denote by $\mathcal{S}$ the cloud server, by $\mathcal{B}$ the subset of honest users, and by $\mathcal{C}$ the subset of users corrupted by the adversary. Also, since we only need to prove the security for an iteration, we omit the notation $t$ in our description.
%


\begin{theorem}
Assuming the security of the underlying mask$\-$ing-based secure aggregation and commitment techniques, \main \enspace ensures the confidentiality of the gradient vectors of individual honest users in the system.
\end{theorem}

\begin{proof}
The proof is mostly similar to \textit{Theorem 6.3} in \cite{bonawitz2017practical}, which indicates that the masking mechanism in secure aggregation protects the confidentiality of the gradient vectors of individual honest users, due to the existence of a simulator \textsf{SIM} for simulating the masked gradient vectors.
On the other hand, we need to additionally consider simulation for the messages related to the verification.
Firstly, we need to consider the hashes from $\textsf{HF}(\cdot)$ which are committed in the \textit{Making commitments} step. 
Here note that the simulator \textsf{SIM} does not know the real inputs of honest users by the time it needs to compute the hash and the commitment.
For this, it can generate a dummy vector, hash it, and compute the commitment.
Given the security of the masking technique and the hiding property of commitment, the joint view of $\mathcal{C}$ and $\mathcal{S}$ is indistinguishable from that in real protocol execution.


Secondly, we need to consider the verification phase.
In particular, we need to show that the joint view of users in $\mathcal{C}$ and $\mathcal{S}$ is indistinguishable from that in the real protocol execution.
The subtlety here is that \textsf{SIM} commits to dummy hashes in the beginning, which are different from the hashes of vectors sampled by \textsf{SIM} after seeing the aggregation result.
Fortunately, due to the equivocal property of commitment, in the common reference string (CRS)-hybrid model \cite{Lindell17}, \textsf{SIM} can obtain a trapdoor for the commitment scheme, which can be used to equivocate the simulated commitments to the hashes of vectors sampled by it on behalf of honest users, based on the aggregation result of honest users. 
The simulated hashes in the verification phase thus can successfully pass the commitment verification, followed by the aggregation result integrity verification.

\end{proof}

\begin{theorem}
Assuming the security of the underlying homomorphic hash function and commitment techniques, VP$\-$FedMF ensures the integrity of aggregation on the server side. In particular, in a certain iteration, an honest user will accept the received updated item vector $\mathbf{v}_k$ derived from aggregation if and only if it is correctly produced by the server.
\end{theorem}

\begin{proof}
Assume that there exists a probabilistic polynomial-time (PPT) adversary which can produce a forged aggregation result $\mathbf{v}_k^*$ ($\mathbf{v}_k^*\neq \mathbf{v}_k$), and make an honest user $\mathcal{U}_i\in \mathcal{B}$ accept the forged aggregation result.
Firstly, since $\mathcal{U}_i$ does not output $\perp$, the decommitment strings from the users in $\mathcal{C}$ should be able to pass the commitment verification phase.
Here, it is noted that due to the binding property of the commitment technique, the commitment verification will fail with a non-negligible probability if the users in $\mathcal{C}$ instructed by the adversary send malformed decommitment strings.
So once the hash values have been committed, the users in $\mathcal{C}$ cannot change them without having an honest user $\mathcal{U}_i$ output $\perp$.
If the adversary manages to have an honest user accept the forged aggregation result $\mathbf{v}_k^*$, it is required that $\textsf{HF}(\mathbf{v}_k)=\textsf{HF}(\mathbf{v}_k^*)$, i.e., $$\prod\nolimits_{l \in [1,d]} g_l^{v_{k,l}}=\prod\nolimits_{l \in [1,d]} g_l^{v^*_{k,l}}.$$

However, given that $\mathbf{v}_k^*\neq \mathbf{v}_k$, this will happen with negligible probability, given the collision resistance property of the homomorphic hash function.
Therefore, the assumption in the beginning does not hold.
The adversary cannot have an honest user accept a forged aggregation result in \main.

\end{proof}

\section{Experiments}\label{sec:experiment}

\subsection{Setup}
We implement VPFedMF in Python. In particular, the homomorphic hash function $\textsf{HF}(\cdot)$ is realized via elliptic curve NIST-P256. The commitment scheme is realized via hash commitments instantiated via SHA-256.
%
%
For pseudo-random number generator, we use AES in CTR mode. For key agreement, we use Diffie-Hellman key exchange over elliptic curve NIST-P256. In addition, we set the modulus $B=2^{34}$.
We use a real-world movie rating dataset MovieLens \cite{harper2015movielens}, which consists of 610 users rating on 9712 movies.
We adopt a common trick for scaling floating-point numbers up to integers as required by cryptographic computation \cite{WangRWW13,ZhengDW18}, where a large scaling factor $\alpha=10^7$ is used.
The server process and user process are deployed on a laptop equipped with a 4-core Intel i5-8300H CPU (2.3 GHz) and 8 GB RAM. 
For running-time related experiments, we report the results averaged over 10 runs.
In our experiments, we compare with the state-of-the-art prior work by Chai \textit{et al.} \cite{chai2020secure}.

\subsection{Offline Optimization}

VPFedMF aims to be utilized in a setting where multiple users want to collaboratively train a joint model for personalized recommendation so as to benefit each other, while keeping their privacy preserved.
%
%
%
Hence, we consider all users are willing to participate in each iteration. 
Namely, in the setting considered by VPFedMF, the participants are not limited with computation resource or network bandwidth, as opposed to the IoT setting. 
In such context, we perform the following offline processing for performance optimization.
Recall that the computation of the homomorphic hash function is within the cyclic group $\mathbb{G}$, which needs to produce the element $g_l^{x_l}$ in each dimension of the input vector $\mathbf{x}$ via expensive exponentiation.
In order to circumvent the latency from such expensive computation in the group, our idea is to pre-generate a set of group elements in an offline phase.
When the actual learning process takes place, the computation of $g_l^{x_l}$ can be simply converted to the fast searching over a set of elements.

\begin{table*}[t!]
  \large
  \centering
  \caption{VPFedMF's Computation Performance in the PartText Setting}
  \label{tab:computation overhead PartText}%
  \begin{spacing}{0.8}
  \resizebox{0.65\linewidth}{!}{
    
    \begin{tabular}{cccccccccc}
    \toprule
    \multicolumn{1}{c}{\multirow{4}[2]{*}{User}} & \multicolumn{1}{c}{\multirow{4}[2]{*}{Items}} & \multicolumn{1}{c}{\multirow{4}[2]{*}{}} & \multirow{4}[2]{*}{Phase 1} & \multicolumn{3}{c}{\multirow{2}{*}{Phase 2 }} & \multicolumn{3}{c}{\multirow{2}{*}{Phase 3}} \\
          &       & \multicolumn{1}{c}{} & \multicolumn{1}{c}{} & \multicolumn{3}{c}{}          & \multicolumn{3}{c}{} \\
\cmidrule{5-10}          &       & \multicolumn{1}{c}{} & \multicolumn{1}{c}{} & \multicolumn{1}{c}{\multirow{2}{*}{0}} & \multicolumn{1}{c}{\multirow{2}{*}{1}} & \multicolumn{1}{c}{\multirow{2}{*}{2}}  & \multicolumn{1}{c}{\multirow{2}{*}{0}} & \multicolumn{1}{c}{\multirow{2}{*}{1}}  & \multicolumn{1}{c}{\multirow{2}{*}{2}}\\
          &       & \multicolumn{1}{c}{} & \multicolumn{1}{c}{} & \multicolumn{1}{c}{} & \multicolumn{1}{c}{} & \multicolumn{1}{c}{} & \multicolumn{1}{c}{} & \multicolumn{1}{c}{} & \multicolumn{1}{c}{}\\
    \midrule
    \multirow{6}[6]{*}{100} & \multirow{2}{*}{60} & User & 1 ms  & 14 ms & 26 ms & –      & 0 ms     & 24 ms  & 233 ms  \\
          &       & Server & –     & 1 ms   & –   & 53 ms     & 1 ms  & –   & –  \\
\cmidrule{2-10}          & \multirow{2}{*}{240} & User & 3 ms  & 31 ms & 58 ms & –   & 0 ms     & 54 ms  & 745 ms \\
          &       & Server & –     & 2 ms & –    & 178 ms     & 7 ms & –    & –  \\
\cmidrule{2-10}         & \multirow{2}[2]{*}{640} & User & 7 ms & 57 ms & 78 ms & –  & 0 ms    & 133 ms & 1942 ms  \\
          &       & Server & –     & 11 ms  & –  & 273 ms     & 18 ms & –  & – \\
    \midrule
    \multirow{6}[6]{*}{300} & \multirow{2}[2]{*}{60} & User & 2 ms & 14 ms & 68 ms & – & 0 ms     & 50 ms  & 460 ms \\
          &       & Server & –     & 3 ms  & –   & 156 ms     & 5 ms & –   & – \\
\cmidrule{2-10}          & \multirow{2}[2]{*}{240} & User & 6 ms & 35 ms & 148 ms & –  & 0 ms     & 157 ms & 1288 ms \\
          &       & Server & –     & 11 ms  & –  & 440 ms     & 20 ms & – & –  \\
\cmidrule{2-10}          & \multirow{2}[2]{*}{640} & User & 9 ms & 60 ms & 198 ms  & –   & 0 ms & 285 ms & 2501 ms\\
          &       & Server & –     & 28 ms  & –   & 717 ms     & 49 ms & –   & – \\
    \bottomrule
    \end{tabular}%
   }

   \end{spacing}

\end{table*}%

\begin{table*}[t!]
  \large
  \centering
  \caption{VPFedMF's Computation Performance in the FullText Setting}
  \label{tab:computation overhead FullText}%
  \begin{spacing}{0.8}
  \resizebox{0.65\linewidth}{!}{
    \begin{tabular}{cccccccccc}
    \toprule
    \multicolumn{1}{c}{\multirow{4}[2]{*}{User}} & \multicolumn{1}{c}{\multirow{4}[2]{*}{Items}} & \multicolumn{1}{c}{\multirow{4}[2]{*}{}} & \multirow{4}[2]{*}{Phase 1} & \multicolumn{3}{c}{\multirow{2}{*}{Phase 2 }} & \multicolumn{3}{c}{\multirow{2}{*}{Phase 3}} \\
          &       & \multicolumn{1}{c}{} & \multicolumn{1}{c}{} & \multicolumn{3}{c}{}          & \multicolumn{3}{c}{} \\
\cmidrule{5-10}          &       & \multicolumn{1}{c}{} & \multicolumn{1}{c}{} & \multicolumn{1}{c}{\multirow{2}{*}{0}} & \multicolumn{1}{c}{\multirow{2}{*}{1}} & \multicolumn{1}{c}{\multirow{2}{*}{2}}  & \multicolumn{1}{c}{\multirow{2}{*}{0}} & \multicolumn{1}{c}{\multirow{2}{*}{1}}  & \multicolumn{1}{c}{\multirow{2}{*}{2}}\\
          &       & \multicolumn{1}{c}{} & \multicolumn{1}{c}{} & \multicolumn{1}{c}{} & \multicolumn{1}{c}{} & \multicolumn{1}{c}{} & \multicolumn{1}{c}{} & \multicolumn{1}{c}{} & \multicolumn{1}{c}{}\\
    \midrule
    \multirow{6}[6]{*}{100} & \multirow{2}{*}{60} & User & 3 ms  & 40 ms & 312 ms & –      & 0 ms     & 68 ms  & 564 ms  \\
          &       & Server & –     & 1 ms   & –   & 204 ms     & 1 ms  & –   & –  \\
\cmidrule{2-10}          & \multirow{2}{*}{240} & User & 13 ms  & 168 ms & 1245 ms & –   & 0 ms     & 228 ms  & 2171 ms \\
          &       & Server & –     & 3 ms & –    & 763 ms     & 7 ms & –    & –  \\
\cmidrule{2-10}         & \multirow{2}[2]{*}{640} & User & 75 ms & 454 ms & 3211 ms & –  & 0 ms    & 614 ms & 5607 ms  \\
          &       & Server & –     & 10 ms  & –  & 2001 ms     & 20 ms & –  & – \\
    \midrule
    \multirow{6}[6]{*}{300} & \multirow{2}[2]{*}{60} & User & 42 ms & 41 ms & 892 ms & – & 0 ms     & 169 ms  & 1501 ms \\
          &       & Server & –     & 3 ms  & –   & 640 ms     & 5 ms & –   & – \\
\cmidrule{2-10}          & \multirow{2}[2]{*}{240} & User & 93 ms & 172 ms & 3785 ms & –  & 0 ms     & 779 ms & 5647 ms \\
          &       & Server & –     & 11 ms  & –  & 2556 ms     & 32 ms & – & –  \\
\cmidrule{2-10}          & \multirow{2}[2]{*}{640} & User & 168 ms & 459 ms & 10261 ms  & –   & 0 ms & 1959 ms & 15132 ms\\
          &       & Server & –     & 26 ms  & –   & 6415 ms     & 70 ms & –   & – \\
    \bottomrule
    \end{tabular}%
   }

   \end{spacing}
\end{table*}

\subsection{Computation Overhead of Each Iteration}

In this section, we first analyze fine-grained time consumption in each step of a single iteration when training over the Movielens dataset in VPFedMF. Then we evaluate the computation overhead for each iteration and compare with the results reported in FedMF~\cite{chai2020secure}. In addition, time consumption as a function of the dimension size (i.e., the dimension $d$ of the latent user profile and item profile) is evaluated. 

Following FedMF~\cite{chai2020secure}, we evaluate two rating settings: \textit{PartText} and \textit{FullText}. These two settings have slight difference when users submit vectors to the server. In the \textit{PartText} setting, users are allowed to only upload gradient vectors for items which have been rated. As for the \textit{FullText} setting, users submit gradient vectors from all items.
For items that a user has not rated, the corresponding elements in the gradient vector are set to 0.

\subsubsection{Fine-Grained Time Consumption of Each Step}\label{sec:computation}

The computation overhead for each protocol step as detailed in Fig.~\ref{fig:Privacy security protocol} is comprehensively evaluated and summarized in Table \ref{tab:computation overhead PartText} with respect to \textit{PartText} and Table \ref{tab:computation overhead FullText} with respect to \textit{FullText}. Specifically, we fix the number of users to be 100 and 300 in both \textit{PartText} and \textit{FullText} settings but vary the number of items for all users to evaluate computation performance. Note that the computation overhead of users reported in this work is actually the average time consumption in each step per user.

As we can see from Table \ref{tab:computation overhead PartText} and Table \ref{tab:computation overhead FullText}, most time is consumed in secure aggregation phase for the reason that each user needs to utilize $\textrm{ck}_{i,j}$, item id $k$ and iteration $t$ as the input for the \textsf{PRNG} on the user side while the server needs to aggregate all these masked vectors for each item $\mathcal{V}_k$ rated by $n_k$ users. Besides, for the verification phase, in the $\textit{Aggregation result verification}$ step each user $\mathcal{U}_i$ needs to verify all the updated item vectors $\mathbf{v}_k^t$ from the server based on $\textsf{HF}(\cdot)$.
Consider 300 users and 640 rated items as an example. The computation overhead for a user in the \textit{Masking gradient vectors} step is 198 ms in \textit{PartText} and 10261 ms in \textit{FullText}. Besides, at the end of the iteration, each user needs to spend 285 ms and 2501 ms in \textit{PartText}, 1959 ms and 15132 ms in \textit{FullText}, for the \textit{Commitment verification} and \textit{Aggregation result verification} steps respectively. These three steps dominate the whole overall consumption in each iteration on the user side. As for the server, time consumption in the \textit{Aggregating masked gradient vectors} step takes up over 90\% (exactly 717 ms in \textit{PartText} and 6415 ms in \textit{FullText}). Consequently, time consumption by users and the server induced by other steps can be comparatively neglected.

\subsubsection{Overall Time Consumption and Comparison}

\begin{table}[b!]
\centering
\caption{Different Rating Settings for Training}
\small
\setlength{\tabcolsep}{1.5mm}{
\begin{spacing}{0.6}
\begin{tabular}{@{}cccccccc@{}}

\toprule
Items & 60    & 80    & 160   & 320   & 640   & 1280  & 2560  \\ \midrule
Ratings & 9497  & 12087 & 20512 & 32371 & 47883 & 65728 & 81786\\ \bottomrule
\label{Table:time consumption for each iteration}
\end{tabular}
\end{spacing}}
\end{table}

\begin{figure}[t!]
\centerline{\includegraphics[width=0.35\textwidth]{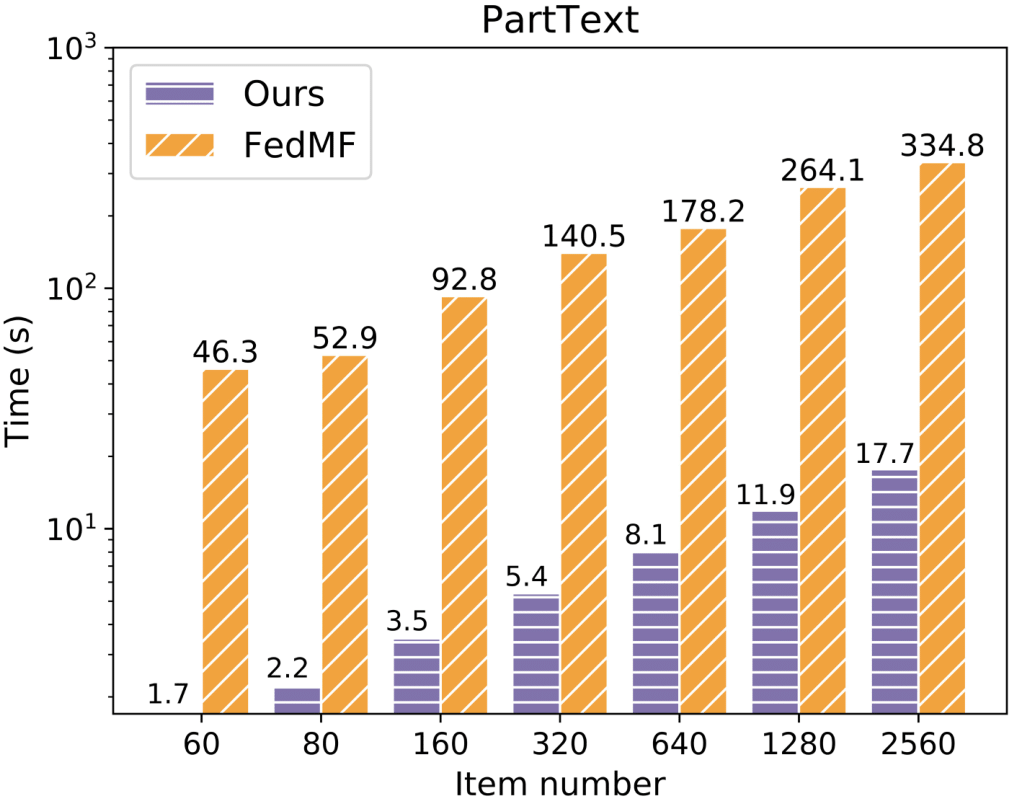}}
\caption{Time consumption of each iteration as a function of number of items held by each user.}
\label{fig:time consumption for each iteration Part}
\end{figure}

\begin{figure}[t!]
\centerline{\includegraphics[width=0.35\textwidth]{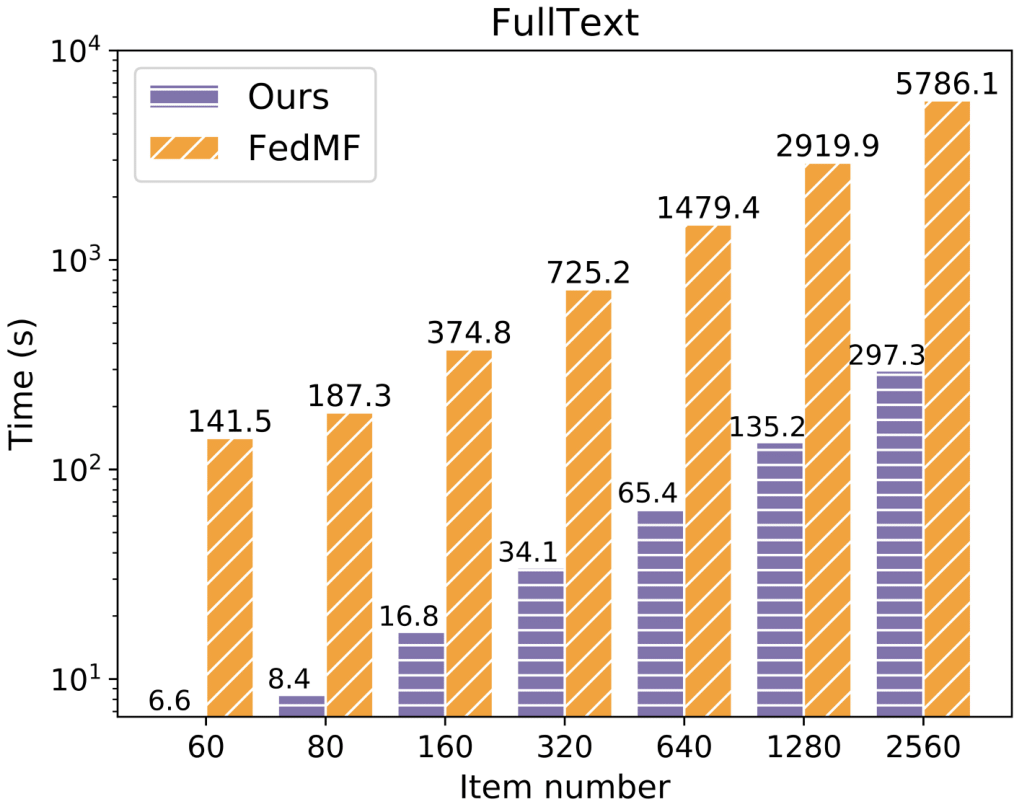}}
\caption{Time consumption of each iteration as a function of number of items held by each user.}
\label{fig:time consumption for each iteration Full}
\end{figure}

In FedMF \cite{chai2020secure}, gradient updates are protected by additive homomorphic encryption, and the aggregation is performed over the resulting ciphertexts.
%
Although it can thwart privacy data leakage, homomorphic encryption is too costly to be efficient enough in practice. 
%
To make an apple-to-apple comparison, we follow the same setting as \cite{chai2020secure}, where the number of users is fixed to 610 for training. The dimension $d$ is set to 100. Table~\ref{Table:time consumption for each iteration} summarizes the varying number of ratings in our implementation. 

As shown in Fig. \ref{fig:time consumption for each iteration Part} and Fig. \ref{fig:time consumption for each iteration Full}, the time consumption of VPFedMF for each iteration is significantly less than the counterpart FedMF, under both \textit{PartText} and \textit{FullText} settings. More specifically, the time consumption in FedMF is about $20\times $ higher than the VPFedMF with the item number varying from 60 to 2560. For example, VPFedMF costs 297.3 seconds in the \textit{FullText} setting to train 2560 items, compared with 5786.1 seconds in FedMF. As for the \textit{PartText} setting,  VPFedMF only costs 17.7 seconds when training 2560 items, in contrast to 334.8 seconds in FedMF.

\begin{table*}[t!]
  \large
  \centering
  \begin{spacing}{1.0}
   \caption{Outgoing communication overhead for each step in the PartText setting}
  \label{tab:PartText outgoing communication overhead}
  \resizebox{0.65\linewidth}{!}{
    \begin{tabular}{cccccccccc}
    \toprule
    \multicolumn{1}{c}{\multirow{4}[2]{*}{Users}} & \multicolumn{1}{c}{\multirow{4}[2]{*}{Items}} & \multicolumn{1}{c}{\multirow{4}[2]{*}{}} & \multirow{4}[2]{*}{Phase 1} & \multicolumn{3}{c}{\multirow{2}{*}{Phase 2}} & \multicolumn{3}{c}{\multirow{2}{*}{Phase 3}} \\
          &       & \multicolumn{1}{c}{} & \multicolumn{1}{c}{} & \multicolumn{3}{c}{}          & \multicolumn{3}{c}{} \\
\cmidrule{5-10}          &       & \multicolumn{1}{c}{} & \multicolumn{1}{c}{} & \multicolumn{1}{c}{\multirow{2}{*}{0}} & \multicolumn{1}{c}{\multirow{2}{*}{1}} & \multicolumn{1}{c}{\multirow{2}{*}{2}} &  \multicolumn{1}{c}{\multirow{2}{*}{0}} & \multicolumn{1}{c}{\multirow{2}{*}{1}} & \multicolumn{1}{c}{\multirow{2}{*}{2}}\\
          &       & \multicolumn{1}{c}{} & \multicolumn{1}{c}{} & \multicolumn{1}{c}{} & \multicolumn{1}{c}{} & \multicolumn{1}{c}{} & \multicolumn{1}{c}{} & \multicolumn{1}{c}{} & \multicolumn{1}{c}{}
          \\
    \midrule
    \multirow{6}[6]{*}{100} & \multirow{2}[2]{*}{60} & User & –     & 5.09 KB & 131.34 KB & – & 5.23 KB    & – & –  \\
          &       & Server & –     & 146.58 KB& –& 140.63 KB     & 150.95 KB & – & – \\
\cmidrule{2-10}          & \multirow{2}[2]{*}{240} & User & –     & 19.71 KB & 509.13 KB & – & 20.26 KB & –  & –\\
          &       & Server & –     & 399.34 KB & – & 562.50 KB     & 411.18 KB & – & – \\
\cmidrule{2-10}          & \multirow{2}[2]{*}{640} & User & –     & 47.68 KB & 1233.16 KB & – & 49.09 KB     & – & – \\
          &       & Server & –     & 711.12 KB & – & 1500.00 KB     & 732.22 KB & – & – \\
    \midrule
    \multirow{6}[6]{*}{300} & \multirow{2}[2]{*}{60} & User & –     & 5.09 KB & 131.34 KB & – & 5.23 KB     & – & –  \\
          &       & Server & –     & 419.32 KB & – & 140.63 KB     & 431.69 KB & – & – \\
\cmidrule{2-10}          & \multirow{2}[2]{*}{240} & User & –     & 19.71 KB & 509.13 KB & –  & 20.31 KB    & –  & –  \\
          &       & Server & –     & 1174.58 KB & –  & 562.50 KB     & 1209.46 KB & –  & –  \\
\cmidrule{2-10}          & \multirow{2}[2]{*}{640} & User & –     & 47.68 KB & 1233.16 KB & – & 49.09 KB    & – & – \\
          &       & Server & –     & 2062.71 KB & –  & 1500.00 KB     & 2123.61 KB & – & –  \\
    \bottomrule
    \end{tabular}%
  }
  \end{spacing}

\end{table*}%

\begin{table*}[t!]
  \large
  \centering
  \begin{spacing}{1.0}
  \caption{Outgoing communication overhead for each step in the FullText setting}
  \label{tab:FullText outgoing communication overhead}
  \resizebox{0.65\linewidth}{!}{
    \begin{tabular}{cccccccccc}
    \toprule
    \multicolumn{1}{c}{\multirow{4}[2]{*}{Users}} & \multicolumn{1}{c}{\multirow{4}[2]{*}{Items}} & \multicolumn{1}{c}{\multirow{4}[2]{*}{}} & \multirow{4}[2]{*}{Phase 1} & \multicolumn{3}{c}{\multirow{2}{*}{Phase 2}} & \multicolumn{3}{c}{\multirow{2}{*}{Phase 3}} \\
          &       & \multicolumn{1}{c}{} & \multicolumn{1}{c}{} & \multicolumn{3}{c}{}          & \multicolumn{3}{c}{} \\
\cmidrule{5-10}          &       & \multicolumn{1}{c}{} & \multicolumn{1}{c}{} & \multicolumn{1}{c}{\multirow{2}{*}{0}} & \multicolumn{1}{c}{\multirow{2}{*}{1}} & \multicolumn{1}{c}{\multirow{2}{*}{2}} &  \multicolumn{1}{c}{\multirow{2}{*}{0}} & \multicolumn{1}{c}{\multirow{2}{*}{1}} & \multicolumn{1}{c}{\multirow{2}{*}{2}}\\
          &       & \multicolumn{1}{c}{} & \multicolumn{1}{c}{} & \multicolumn{1}{c}{} & \multicolumn{1}{c}{} & \multicolumn{1}{c}{} & \multicolumn{1}{c}{} & \multicolumn{1}{c}{} & \multicolumn{1}{c}{}
          \\
    \midrule
    \multirow{6}[6]{*}{100} & \multirow{2}[2]{*}{60} & User & –     & 5.45 KB & 140.63 KB & – & 5.63 KB    & – & –  \\
          &       & Server & –     & 539.47 KB& – & 140.63 KB     & 555.49 KB & – & – \\
\cmidrule{2-10}          & \multirow{2}[2]{*}{240} & User & –     & 21.80 KB & 562.50 KB & – & 22.48 KB & –  & –\\
          &       & Server & –     & 2157.88 KB & – & 562.50 KB     & 2221.93 KB & – & – \\
\cmidrule{2-10}          & \multirow{2}[2]{*}{640} & User & –     & 58.13 KB & 1500.00 KB & – & 59.91 KB     & – & – \\
          &       & Server & –     & 5754.37 KB & – & 1500.00 KB     & 5924.74 KB & – & – \\
    \midrule
    \multirow{6}[6]{*}{300} & \multirow{2}[2]{*}{60} & User & –     & 5.45 KB & 140.63 KB & – & 5.63 KB     & – & –  \\
          &       & Server & –     & 1629.31 KB & – & 140.63 KB     & 1677.55 KB & – & – \\
\cmidrule{2-10}          & \multirow{2}[2]{*}{240} & User & –     & 21.80 KB & 562.50 KB & –  & 22.48 KB    & –  & –  \\
          &       & Server & –     & 6517.26 KB & –  & 562.50 KB     & 6709.67 KB & –  & –  \\
\cmidrule{2-10}          & \multirow{2}[2]{*}{640} & User & –     & 58.13 KB & 1500.00 KB & – & 59.91 KB    & – & – \\
          &       & Server & –     & 17379.37 KB & –  & 1500.00 KB     & 17893.42 KB & – & –  \\
    \bottomrule
    \end{tabular}%
  }
  \end{spacing}

\end{table*}%

\subsubsection{Scalability with Respect to the Dimension}

Matrix factorization decomposes the sparse rating matrix $\mathbf{R}$ into the user profile matrix and the item profile matrix, where each row vector in both matrices is of the same latent dimension $d$. Time consumption differs for varying latent dimension $d$ due to various sizes of user profile matrix and item profile matrix. To evaluate how the dimension size $d$ affects the time consumption of VPFedMF, we fix the user number to be 610 and item number to be 320, and vary the dimension $d$ for evaluation.
The results are detailed in Fig.~\ref{fig:PartText dimension time consumption for each iteration} and Fig.~\ref{fig:FullText dimension time consumption for each iteration}. With the dimension increasing, the computation cost of the whole system increases approximately \textit{linearly}. Compared to the \textit{FullText} setting, each iteration exhibits less time consumption in \textit{PartText} setting (about $6\times$ to $7 \times$ in our experiments).

%


\begin{figure}[t!]
\centerline{\includegraphics[width=0.35\textwidth]{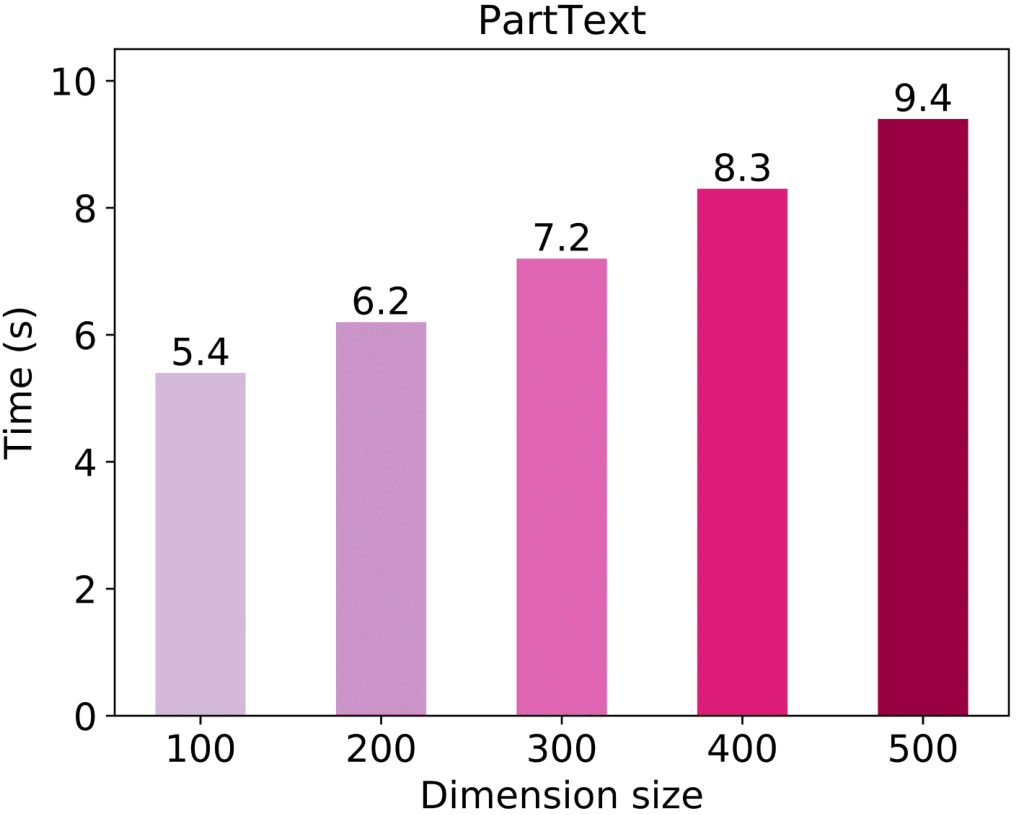}}
\caption{Computation overhead under the \textit{PartText} setting as a function of dimension size $d$.}
\label{fig:PartText dimension time consumption for each iteration}
\end{figure}

\begin{figure}[t!]
\centerline{\includegraphics[width=0.35\textwidth]{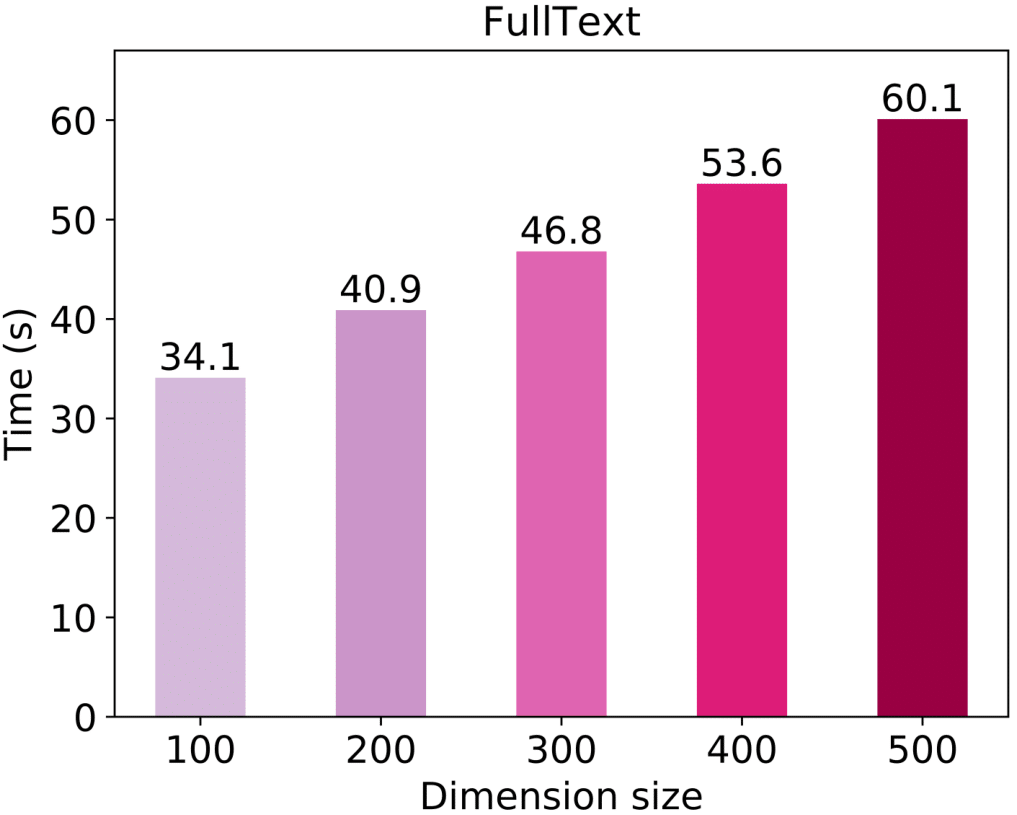}}
\caption{Computation overhead under the \textit{FullText} setting as a function of dimension size $d$.}
\label{fig:FullText dimension time consumption for each iteration}
\end{figure}

\subsection{Outgoing Communication Overhead}

We evaluate the outgoing communication overhead for both \textit{PartText} and \textit{FullText} settings. Particularly, the overhead from the server is outgoing communication sent from the server to a single user. When concerning on the communication channel from user to server, we measure the maximum sizes of the packets transmitted by a particular user as the communication overhead in each step on the user side. Table \ref{tab:PartText outgoing communication overhead} and Table \ref{tab:FullText outgoing communication overhead} summarize communication overhead for each iteration in our VPFedMF under \textit{PartText} and \textit{FullText} settings, respectively. 
The setup is same as Section \ref{sec:computation}. In the secure aggregation phase, most communication consumption is spent on the \textit{Masking gradient vectors} step on the user side, where each user needs to send all the masked gradient vectors to the server.
As for the server, it spends most on the \textit{Making commitments} and \textit{Aggregating masked gradient vectors} steps. 
When the user number is fixed, the communication overhead grows approximately \textit{linearly} in \textit{FullText} with the increasing item number. Note that between the two settings, the gap of communication cost in \textit{Making commitments} and \textit{Decommitting} steps on the server side enlarges when item number increases.
This is because that in the \textit{FullText} setting, the server needs to broadcast the \{$c^t_{i,k}$\}, \{$h^t_{i,_k},r^t_{i,k}$\} for all items to users even given that the rating matrix is sparse, while in the \textit{PartText} setting the server only needs to broadcast them for the rated items provided by corresponding users.
%


\subsection{Accuracy}

\begin{figure}[t!]
\centerline{\includegraphics[width=0.35\textwidth]{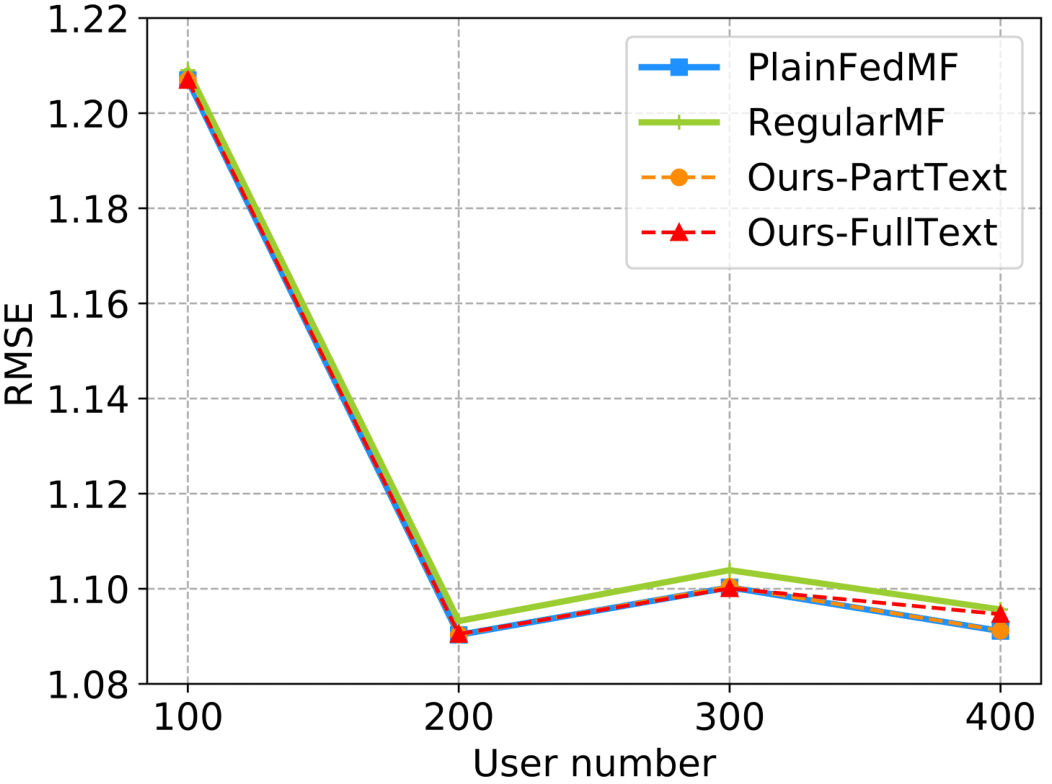}}
\caption{RMSE as a function of varying user numbers.}
\label{fig:RMSE users}
\end{figure}

\begin{figure}[t!]
\centerline{\includegraphics[width=0.35\textwidth]{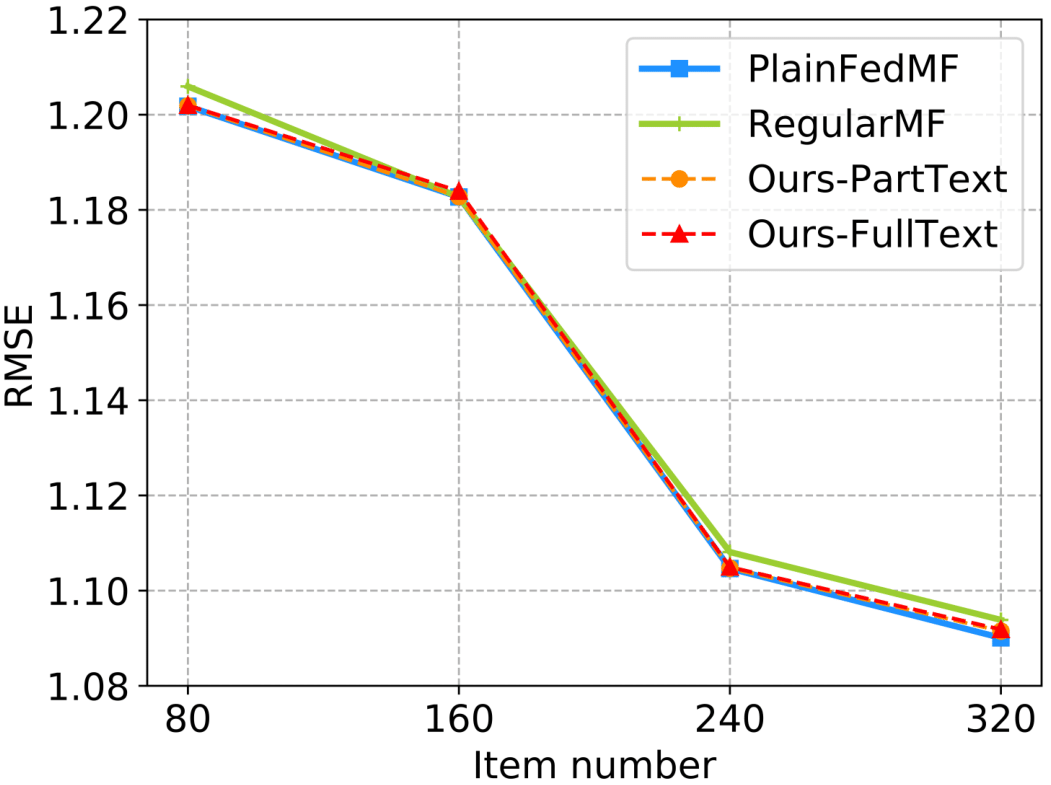}}
\caption{RMSE as a function of varying item numbers.}
\label{fig:RMSE items}
\end{figure}

Root Mean Squared Error (RMSE) is a common accuracy metric used in recommender systems to evaluate the training performance \cite{koren2009matrix}. We utilize RMSE to examine the accuracy of VPFedMF in both \textit{PartText} and \textit{FullText} settings, which is compared with FedMF, PlainFedMF---federated MF in plaintext domain as described in Algorithm~\ref{DistributedMF}---and the conventional centralized MF abbreviated as RegularMF.
We set the iteration number to 50 so that the training processes of these four schemes converge. In Fig. \ref{fig:RMSE users} we fix the item number to be 300 and in Fig. \ref{fig:RMSE items} we fix the user number to be 300. Both figures illustrate the RMSE of each aforementioned MF scheme by varying the number of users and items, respectively. These four schemes show almost the same RMSE with negligible gap.

\section{Conclusion}\label{sec:conclusion}

In this paper, we propose VPFedMF, a new protocol for privacy-preserving and verifiable federated matrix factorization.
VPFedMF provides protection for the individual gradient updates through masking-based lightweight secure aggregation, which allows the server to perform aggregation to update the item profile matrix without seeing individual gradient updates.
In the meantime, VPFedMF allows users to have cryptographic verification on the correctness of the aggregation result produced by the server in each iteration, building on techniques including homomorphic hash function and commitment.
VPFedMF is tested over a real-world movie rating dataset for federated matrix factorization.
The evaluation results demonstrate the practicality of VPFedMF, as well as the performance advantage over prior art (in addition to the security advantage).

\section{Acknowledgment}

This work was supported in part by the National Natural Science Foundation of China (Grants 62002167 and 61702268), by the National Natural Science Foundation of JiangSu (Grant BK20200461), by the Shenzhen Science and Technology Program (Grant RCBS20210609103056041), and by the Guangdong Basic and Applied Basic Research Foundation (Grant 2021A1515110027).

This work was initialized and partially done when X. Wan was with Nanjing University of Science and Technology and mentored by Y. Gao.

\bibliography{References}

\end{document}